\documentclass[reqno,oneside,12pt]{amsart}

\title[Interatomic Force Constants in \tiny{r}\scriptsize{HF}]{Decay of Interatomic Force Constants\\in the Reduced Hartree--Fock model}
\author{\'Eric Canc\`es}
\author{Antoine Levitt}
\author{Jack Thomas}
\date{\today}

\newcommand{\theaddresses}{\noindent\texttt{eric.cances@enpc.fr} \\ ---
CERMICS, Ecole des Ponts - Institut Polytechnique de Paris and Inria, 6-8 avenue Blaise Pascal, Cit\'e Descartes, 77455 Marne-la-Vall\'ee, France.
    \\
    \texttt{antoine.levitt@universite-paris-saclay.fr}, \texttt{jack.thomas@universite-paris-saclay.fr} \\ --- Universit\'e Paris-Saclay, CNRS, Laboratoire de math\'ematiques d'Orsay, 91405, Orsay, France.
}

\usepackage{ifthen, xcolor,color, bm, braket, amsmath, amssymb, mathtools, mathrsfs, microtype, enumitem, cite, tikz,esint,forloop,subfigure,stmaryrd, amssymb,dsfont, graphicx}
\usepackage{caption}
\usepackage[hang,flushmargin]{footmisc}

\usepackage{hyperref}
\hypersetup{colorlinks=true, allcolors=blue}

\usepackage{soul}
\usepackage{yhmath}
\usepackage{textpos}
\usepackage{datetime}
 \newdateformat{mydate}{\ordinal{DAY}~\shortmonthname[\THEMONTH]~\THEYEAR}
 \mydate

 \DeclareFontFamily{U}{mathx}{}
 \DeclareFontShape{U}{mathx}{m}{n}{<-> mathx10}{}
 \DeclareSymbolFont{mathx}{U}{mathx}{m}{n}
 \DeclareMathAccent{\widehat}{0}{mathx}{"70}
 \DeclareMathAccent{\widecheck}{0}{mathx}{"71}

\newcommand{\dos}{{\rm DOS}}
\renewcommand{\leq}{\leqslant}
\renewcommand{\geq}{\geqslant}
\renewcommand{\above}[2]{\genfrac{}{}{0pt}{}{#1}{#2}}

\DeclareMathOperator*{\supp}{supp}

\DeclareMathOperator*{\re}{Re}

\DeclareMathOperator*{\Tr}{Tr}
\DeclareMathOperator*{\tr}{Tr}
\DeclareMathOperator*{\den}{den}
\DeclareMathOperator*{\pv}{pv}
\renewcommand{\i}{i}
\newcommand{\D}{\mathrm{d}}

\renewcommand{\t}{\mathsf{T}} 
\newcommand{\ep}{\varepsilon}

\newcommand{\R}{\mathbb{R}}

\newcommand{\N}{\mathbb{N}}

\renewcommand{\L}{\mathbb{L}}

\newtheorem*{fact*}{Fact}
\newtheorem{theorem}{Theorem}[section]
\newtheorem*{theorem*}{Theorem}

\newtheorem{lemma}[theorem]{Lemma}

\theoremstyle{remark}
\newtheorem{rem}{Remark}
\newtheorem{noname}{}
\newtheorem{as}{Assumption}
\newenvironment{remark}{\begin{rem}}{\hfill $\star$ \end{rem}}

\numberwithin{equation}{section}

\usepackage{geometry}
\usepackage{cleveref}
\numberwithin{equation}{section}
\crefformat{equation}{\textup{(#2#1#3)}}
\crefrangeformat{equation}{(#3#1#4)--(#5#2#6)}
\crefmultiformat{equation}{(#2#1#3)}%
{ and~(#2#1#3)}{, (#2#1#3)}{ and~(#2#1#3)}
\crefformat{section}{\S#2#1#3} 
\crefformat{subsection}{\S#2#1#3}
\crefformat{subsubsection}{\S#2#1#3}
\crefformat{prop}{Proposition~#2#1#3}
\crefformat{theorem}{Theorem~#2#1#3}
\crefformat{lemma}{Lemma~#2#1#3}
\crefformat{corollary}{Corollary~#2#1#3}

\newcounter{listcounter}
\newcounter{listcounter2}
\newwrite\tempfile

\begin{document}
\newcommand{\ec}[2][1]{\comment[#1]{EC}{red}{\textup{#2}}}

\newcommand{\al}[2][1]{\comment[#1]{AL}{orange}{\textup{#2}}}

\newcommand{\jt}[2][1]{\comment[#1]{JT}{purple}{\textup{#2}\hyperref[sec:comments]{\,\,}}}
\newcommand{\jtm}[2][0]{\marginpar{\tiny{\comment[#1]{JT}{purple}{#2}}}}

\setcounter{page}{1}

\begin{abstract}
  We study the decay of the interatomic force constants (equivalently,
  the smoothness properties of the dynamical matrix) in perfect
  crystals both at finite electronic temperature, and for insulators
  at zero temperature, within the reduced Hartree--Fock approximation
  (also called Random Phase Approximation). At finite temperature the
  electrons are mobile, leading to exponential decay of the force
  constants. In insulators, there is incomplete screening, leading to
  an algebraic decay of dipole-dipole interaction type.
\end{abstract}

\begin{textblock*}{20cm}(-2cm,-2cm) 
\end{textblock*}

\maketitle

\let\thefootnote\relax\footnote{
    \theaddresses
 }

\setcounter{page}{1}
\thispagestyle{empty}

\vspace{-1cm}

\section{Introduction}

Locality is a key concept in many numerical electronic structure methods for the simulation of molecules and materials. 
The most well-known example is that of ``nearsightedness'', introduced by Kohn \cite{
  Kohn1996
}. 
The mathematical manifestation of this property is rapid off-diagonal decay of the density matrix for insulators or at finite temperature (see for example, \cite{
    Goedecker1998,
    Benzi2013,
    ProdanKohn05,
    Prodan2005
}). 
This property has been exploited to develop linear-scaling Kohn--Sham type models 
\cite{
    Yang1991,
    Goedecker99, 
    BowlerMiyazaki2012, 
    Niklasson2011, 
    Rubensson2011
}. 
Nearsightedness of the density matrix is however \textit{not} sufficient to justify the use of interatomic potentials or multi-scale schemes such as QM/MM methods \cite{
    ChenOrtner16, 
    ChenOrtner17,
    ChenOrtnerWang22
}. Indeed, as noted by 
\cite{
    CsanyiAlbaretMorasPayneDeVita05,
    ChenOrtner17
}, a key requirement in order to implement these multi-scale algorithms is the \textit{locality of forces}; derivatives of the force on atom $I$ with respect to atom $J$ decays rapidly as the distance between atoms $I$ and $J$ increases. 
Rigorous results in this direction include work on the Thomas--Fermi--von Weizs\"acker model \cite{NazarOrtner17} and in tight binding models of varying complexity \cite{
    ChenOrtner16, 
    ChenOrtnerThomas2019:locality, Thomas2020:scTB
}. In fact, these papers demonstrate \textit{energy locality}, a stronger locality property that states that the energy may be decomposed into the sum of site energies depending only on a small atomic neighbourhood of the central atom. This leads to theoretical justification for interatomic potentials \cite{ThomasChenOrtner2022:body-order}, as well as the study of thermodynamic limit models for crystalline defects \cite{
    EhrlacherOrtnerShapeev16,
    ChenLuOrtner18, 
    ChenNazarOrtner19, 
    OrtnerThomas2020:pointdef
}. However, these studies do not address the issue of the long-range
Coulomb interaction.

We focus here on the reduced Hartree--Fock (rHF)
model~\cite{Solovej1991}, also called the Random Phase Approximation
(RPA) model in the physics literature. Although usually not sufficient
to obtain quantitative results on real materials, this model
reproduces a number of important physical effects (shell structure,
screening, insulator/metal behavior, ...). It is very popular in
mathematical physics because it is similar to the widely used
Kohn--Sham (KS) model (the rHF model is obtained from the KS model by
discarding the exchange-correlation potential) while having much nicer
properties. In particular, the rHF model is strictly convex in the
density, so that the ground-state density, if it exists, is unique.
The study of the rHF model for crystals was initiated in the seminal
contribution \cite{Lions2001}, where the existence of a thermodynamic
limit for perfect crystals was proved. The case of an insulating
crystal with a local defect was addressed in~\cite{Cances2008}, and
the existence and uniqueness of the rHF ground-state density for
disordered crystals was established in~\cite{Cances2013-disorder}. The
dielectric response of perfect crystals at the rHF level of
theory was studied in~\cite{Cances2010} in the insulating case and in
\cite{Levitt2020} in the finite-temperature case.

In this paper, we study the decay properties of the interatomic force
constants both at finite temperature and for insulators at zero
temperature. The physical picture is the following. Imagine a perfect
crystal at mechanical equilibrium, and move an atom $J$ of charge
$Z_{J}$ by an infinitesimal amount $d$. Focusing only on nuclei, this
creates a charge dipole of size $Z_{J} d$ at $J$, which in turn
creates an electrostatic potential decaying as $R^{-3}$ on an atom $I$
at distance $R$ of $J$. However, this discussion completely ignores
the reaction of the electrons. In systems at finite temperature,
electrons are mobile and are able to react to the electrostatic
potential created by the nuclei, effectively totally screening the
dipole and resulting in an exponentially decaying force. In
insulators, by contrast, electrons tend to be tightly bound to nuclei.
They can move rigidly with the nuclei, effectively decreasing the
effective charge $Z_I$, as well as polarize, leading to an effective
dielectric constant that is larger than its value in vacuum.

In this paper, we prove this behavior rigorously. We do this by
deriving a formula for the interatomic force constants in terms of the
screened Coulomb operator, which involves the dielectric operator (the linear response in the total potential to a
small defect potential). Then, using ideas from \cite{Levitt2020} for
the finite temperature case and \cite{Cances2008,Cances2010} for
insulators, we are able to analyse the regularity properties of the
dynamical matrix (Fourier dual of the force constants). At finite
temperature, one can show that the dynamical matrix is analytic;
whereas for insulators at zero temperature, we show that the dynamical
matrix has a singularity at zero. Therefore, at finite temperature,
the force constants decay exponentially, whereas, in the insulating
case, algebraic decay follows from taking the inverse Fourier
transform of the singularity explicitly.

The paper is organized as follows: in \cref{sec:setting}, we recall the mathematical structure of the periodic rHF model (\cref{sec:periodic}) and introduce the defect rHF model and define the interatomic force constants (\cref{sec:defect}). In Theorem~\ref{thm:W}, we write the interatomic force constants in terms of the so-called screened Coulomb operator. In \cref{sec:main}, we describe the main results of this paper: Theorem~\ref{thm:finite-temperature} for finite temperature and Theorem~\ref{thm:insulators} for insulators at zero temperature. The proofs of the main results are contained in \cref{sec:translation} to \cref{sec:locality}. 
In \cref{sec:translation}, we introduce the dynamical matrices $D_{ss'}(q)$, a useful tool to derive the decay properties of the interatomic force constants. 
Sections~\ref{sec:pf-metals}~and~\ref{sec:locality} are dedicated to the particular case of finite temperature and insulators at zero temperature, respectively.

\section{Setting}
\label{sec:setting}
\subsection{Perfect crystals in the rHF approximation}
\label{sec:periodic}
We consider a perfect crystal with lattice $\mathbb L \subset \mathbb R^d$
and Wigner--Seitz cell (or any other unit cell) $\Omega$. The cell $\Omega$
contains $N_{\rm at}$ atoms with positions $\tau_{s}$ and atomic charges $Z_{s}$ for $s =1,\dots, N_{\rm at}$, so that there are
$N_{\rm el} = \sum_{s=1}^{N_{\rm at}} Z_{s}$ electrons per unit cell (we ignore
spin throughout this paper). For technical reasons (see Remark \ref{rem:point}) we will consider
smeared nuclei, with a smooth and compactly supported charge distribution
$m : \R^{3} \to \R_{+}$. The nuclear charge distribution is the
$\mathbb L$-periodic function
$$
\rho^{\rm nuc}_{\rm per}(x) := \sum_{R \in \mathbb L} \sum_{s=1}^{N_{\rm at}} Z_s m(x-(R+\tau_s)).
$$

In the rHF framework, an electronic state is characterized by a one-body reduced density matrix (1-RDM) $\gamma \in \mathcal S(L^2(\mathbb R^3))$, a bounded self-adjoint operator on $L^2(\mathbb R^3)$, satisfying $0 \le \gamma \le 1$ in the sense of quadratic forms, and locally trace-class (i.e. such that for all compact subsets $K \subset\mathbb R^3$, the positive operator $\mathds 1_K \gamma \mathds 1_K$ is trace-class). Recall that the density $\den(A)$ associated with a locally trace-class operator $A$ on $L^2(\mathbb R^3)$ is the unique function in $L^1_{\rm loc}(\mathbb R^3)$ such that
$$
\forall\psi \in L^\infty_{\rm c}(\mathbb R^3), \quad \tr(A \psi) = \int_{\mathbb R^3} \den(A) \psi.
$$

The periodic rHF equations are \cite{nier1993schrodinger,Cances2008}
\begin{align}
  \label{per_1}\gamma^0_{\rm per} &= f_{T}\left(H_{\rm per}^0-\mu^0\right), \\
  \label{per_2}-\Delta V^0_{\rm per} &= 4\pi \left( \rho^{\rm nuc}_{\rm per}-\rho^0_{\rm per} \right),\\
  \label{per_3}\rho_{\rm per}^{0} &= {\rm den}(\gamma^{0}_{\rm per}),
\end{align}
where $H_{\rm per}^0 \coloneqq -\tfrac 12 \Delta + V^0_{\rm per}$ is the periodic Hamiltonian, $\rho^0_{\rm per}$ is the density of the 1-RDM $\gamma^0_{\rm per}$, the total Coulomb potential $V^0_{\rm per}$ is $\mathbb L$-periodic, and $f_T$ the Fermi-Dirac occupation function at temperature $T \ge 0$ given by
\begin{align*}
  &f_T(\epsilon) = \frac{1}{1+e^{\frac{\epsilon}{T} }}, \quad && \mbox{if $T>0$ (positive temperature)}, \\
  &f_0(\epsilon) = \mathds 1_{(-\infty,0]}(\epsilon), \quad && \mbox{if $T=0$ (zero temperature)}.
\end{align*} 
The electronic density is constrained to have $N_{\rm el}$ electrons
per unit cell ($\int_{\Omega} \rho_{\rm per}^0 = N_{\rm el}$), which ensures that the \cref{per_2} has a solution.  
The Fermi level $\mu^0$ can be interpreted as the Lagrange multiplier of this charge neutrality constraint. The electrostatic potential $V^0_{\rm per}$ is unique up to an additive constant; we impose uniqueness by also requiring that $\int_{\Omega} V^{0}_{\rm per} = 0$. The mean-field Hamiltonian $H^0_{\rm per}$ is an $\mathbb L$-periodic Schr\"odinger operator, and $V^0_{\rm per}$ is at least locally square-integrable. As a consequence, $H^0_{\rm per}$ is self-adjoint on $L^2(\mathbb R^3)$ with domain $H^2(\mathbb R^3)$, and its spectrum is a union of closed bounded intervals (the bands), see Figure~\ref{fig:bands}.

In the rest of this paper, we fix the lattice $\mathbb L$, positions and charges of the $N_{\rm at}$ atoms in the unit cell, and the corresponding solution $\gamma^{0}_{\rm per}, V_{\mathrm{per}}^0, \rho_{\rm per}^0$ to \cref{per_1,per_2,per_3}. Furthermore, we will study both the finite temperature $(T > 0)$ and the zero temperature insulating case:
\begin{as}
  \label{assumption_gap}
  If $T=0$, then $\mu^{0} \notin \sigma(-\tfrac 1 2 \Delta + V^{0}_{\rm per})$.
\end{as}

\begin{figure}[!ht]
    \centering
    \includegraphics[width=.95\textwidth]{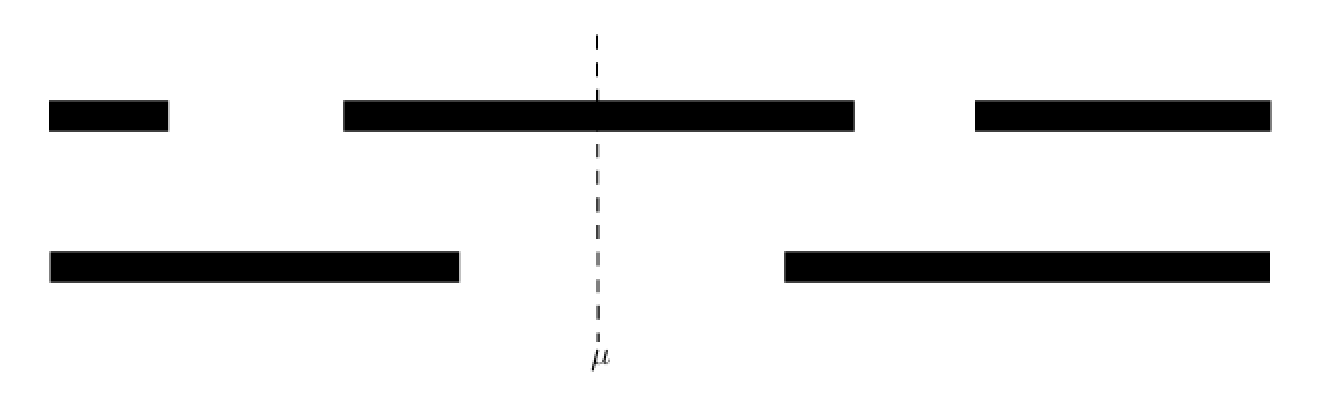}
    \caption{Schematic plots of the spectrum $\sigma(H_{\rm
    per}^0)$ of a metal (top) and
 insulator (bottom).}
    \label{fig:bands}
  \end{figure}

We label by $(R,s)$ the atom located at point $R+\tau_s$, $R \in \L$, $s=1,\dots,N_{\rm at}$.
The $\alpha$-component of the force acting on atom $(R,s)$ is given for $\alpha=1,2,3$ by
\begin{align}
  F_{Rs\alpha} &= -\int_{\R^{3}} Z_s m(x-R-\tau_{s}) \frac{\partial V^{0}_{\rm per}}{\partial x_{\alpha}}(x)\D{x}
  =\int_{\R^{3}} Z_s\frac{\partial m}{\partial x_{\alpha}}(x-R-\tau_{s}) V^{0}_{\rm per}(x)\D{x}.
  \label{eq:FRs}
\end{align}
The crystal is at mechanical equilibrium if $F_{Rs\alpha} = 0$, but we do
not assume this in our results.
\begin{remark}[Hellmann--Feynman]
  For finite systems, the force on an atom is (minus) the
  gradient of the energy with respect to the position of this atom, as can be shown from the
  Hellmann--Feynman theorem. In the crystal case, the energy is
  infinite; however, the force can be obtained as (minus)
  the gradient of the \textit{relative} energy, defined as
  the difference between the energies of the perturbed crystal and the reference crystal in the thermodynamic limit (see \cite{Cances2008}). We
  do not use this formalism here and simply define the $\alpha$-component of the force on atom $(R,s)$ by \cref{eq:FRs}.
\end{remark}

\subsection{The defect problem}
\label{sec:defect}
A small displacement 
$$
d = (d_{Rs})_{\above{R\in \mathbb L,}{s = 1,\dots,N_{\rm at}}} \in \ell^1(\L \times \{1,\dots,N_{\rm at}\};\R^3)
$$
of the atoms from their reference configuration
induces a nuclear charge displacement
\begin{align*}
  \nu^{d} = \sum_{R \in \mathbb L} \sum_{s=1}^{M} Z_s \big[ m\big(\cdot-(R+\tau_{s}+d_{Rs})\big) - m\big(\cdot-(R+\tau_{s})\big)\big] \in L^1(\R^3) \cap L^\infty(\R^3).
\end{align*}
This induces a defect potential $V_{\rm def}^{d} \coloneqq - v_{\rm c} \nu^{d}$, where $v_{\rm c}$ is the \textit{Coulomb operator}, given by
\begin{align*}
  (v_{\rm c} \rho)(x) = \int_{\R^{3}}\frac{\rho(y)}{|x-y|} \D{y}, 
  \qquad \widehat{v_{\rm c} \rho}(q) 
  = 
\frac
    {4\pi}
    {|q|^2}
\widehat{\rho}(q).
\end{align*}
In turn this defect potential induces a reorganization of the
electronic charge, resulting in the self-consistent equation
\begin{align}
  V_{\rm tot}^{d} = V_{\rm def}^{d} + v_{\rm c} \Big( 
  \den f_{T}\big(H_{\rm per}^0 + V_{\rm tot}^{d}  - \mu^{0}\big) 
  - \rho_{\mathrm{per}}^0
  \Big).
  \label{eq:def_rhf}
\end{align}
The existence and uniqueness of solutions of this equation for $d$ small enough in $\ell^1$-norm have been established in the zero-temperature case
in ~\cite{Cances2008,Cances2010} by a variational method, with
$V_{\rm tot}^{d}$ in $\mathcal C' = v_{\rm c}^{1/2} L^{2}(\R^{3})$, the
dual space to the Coulomb space
$\mathcal C = v_{\rm c}^{-1/2} L^{2}(\R^{3})$. In the
finite-temperature case, existence and uniqueness have been
established in \cite{Levitt2020} for $d$ small enough in $\ell^1$ (since $\nu^d$ is small in $H^{-2}$) by a perturbation argument, with $V_{\rm tot}^{d} \in L^{2}$.

Therefore, for both insulators at zero temperature and in the finite temperature case, we may define the force acting on atom $(R,s)$ as 
\begin{align}
  F_{Rs}^{d} &\coloneqq \int_{\R^{3}} Z_{s} \nabla m(x-R-\tau_{s}-d_{Rs}) \,
  \big(V_{\rm per}^{0}(x) + V_{\rm tot}^{d}(x)\big) \D{x}.
  \label{eq:forces}
\end{align}
The central object of our study are the \textit{interatomic force constants}, defined as
\begin{align}
  \boxed{ C_{Rs,R's'}  = \left.\frac{\partial F_{Rs}^d}{\partial d_{R's'}}\right|_{d=0} \in \mathbb R^{3\times3}.}
  \label{eq:IFC}
\end{align}

We show that this is indeed a well-defined object and give a formula for the interatomic force constants in terms of the so-called \textit{screened Coulomb operator}:
\begin{theorem}
\label{thm:W}
    Suppose that Assumption~\ref{assumption_gap} is satisfied. 
    Then, the mapping $d \mapsto F_{Rs}^{d}$ is well-defined in a neighborhood of $0$ in $\ell^1(\L \times \{1,\dots,N_{\rm at}\};\R^3)$, differentiable at $0$ and 
    we have
  \begin{equation}\label{eq:CIJ_ij}
    \boxed{ \big[ C_{Rs,R's'} \big]_{\alpha\beta} =\delta_{RR'}\delta_{ss'} [c_{s}]_{\alpha\beta} + Z_s Z_{s'} \bigg\langle   \frac{\partial m}{\partial x_\alpha} (\cdot - R- \tau_s),   
      W \!\left(  \frac{\partial m}{\partial x_\beta}(\cdot - R'-\tau_{s'}) \!\right)  \!\!\bigg\rangle_{\!\! L^2}},
  \end{equation}
  where $c_s \coloneqq - Z_s \int_{\mathbb R^3} \nabla^2m(\cdot - \tau_s) V_{\mathrm{per}}^0$ and the \emph{screened Coulomb
  operator} $W$ 
  is a bounded linear operator $H^{-2} \to L^2$ for $T>0$, and $\mathcal C \to \mathcal C'$ for $T=0$. 
\end{theorem}

\begin{proof}
  In both cases, it is shown \cite{Cances2008,Cances2010,Levitt2020} that
    \begin{align}
      V_{\rm tot}^{d} 
      = -v_{\rm c} \nu^{d} + v_{\rm c}\chi_0 V_{\rm tot}^{d}
      + O(\|\nu^{d}\|^{2}_{Y})
      \label{eq:V(Vdef)-insu}\; \text{ in } X
    \end{align}
    with $X = \mathcal C', Y = \mathcal C$ in the zero-temperature case and $X=L^2, Y = H^{-2}$ in the finite-temperature case, and
    where the independent-particle susceptibility operator $\chi{_0}$
    is the differential of the map
    $V \mapsto \den f_{T}\big(H_{\rm per}^0 + V - \mu^{0}\big)$ taken
    at $V=0$. Since $\|\nu^d\|_Y \lesssim \|d\|_{\ell^1}$, it follows that
    \begin{align*}
      V_{\rm tot}^{d} = - W \nu^{d} + O(\|d\|_{\ell^1}^{2})
    \end{align*}
    in $X$, with $W = \epsilon^{-1} v_{\rm c}$ the screened Coulomb operator
    and $\epsilon = 1 - v_{\rm c}\chi_0$ the dielectric operator. It
    follows that $d \mapsto F^d$ is differentiable at $0$, with
    \begin{align}
      \frac
      {\partial V_{\rm tot}}
      {\partial d_{R's'\beta}}
      = - W \frac
      {\partial \nu^d}
      {\partial d_{R's'\beta}}
      = Z_{s'} 
      W \frac
      {\partial m}
      {\partial x_{\beta}}
      (\,\cdot - R' - \tau_{s'} - d_{R's'} ).
    \end{align}
    Therefore, for $(R,s) \not= (R',s')$ we have
    \begin{align}
      C_{Rs\alpha,R's'\beta}
      &= Z_s\int_{\R^{3}} 
      \frac
      {\partial m}
      {\partial x_{\alpha}}(x-R-\tau_{s})
      \,
      \frac
      {\partial V_{\rm tot}}
      {\partial d_{R's'\beta}}
      \bigg|_{d = 0}(x) \D{x} \nonumber\\
      &=  Z_sZ_{s'} \int_{\R^{3}} 
      \frac
      {\partial m}
      {\partial x_{\alpha}}(\,\cdot -R-\tau_{s})
      \,
      W
      \frac
      {\partial m}
      {\partial x_{\beta}}
      (\,\cdot - R' - \tau_{s'}),
    \end{align}
as required.
\end{proof}
\begin{remark}
  \label{rem:point}
  We chose to study smeared nuclei for technical simplicity. With
  point nuclei ($m = \delta$), while the formula \eqref{eq:CIJ_ij} is
  still well-defined, the perturbation to first order
  $\partial{}m/\partial{}x_\beta{}$ is a derivative of the Dirac
  distribution, which is not $\Delta{}$-bounded, so that the usual
  methods of proving the convergence of Rayleigh--Schrödinger
  perturbation theory (and, \textit{a fortiori}, Theorem \ref{thm:W}) are
  inapplicable. One can in this case use the method of
  \cite{hunziker1986distortion}, which unitarily maps the Schrödinger
  operator with variable nuclei positions to a Schrödinger operator
  with fixed nuclei positions but variable metric, to which regular
  perturbation theory is applicable.
\end{remark}

\section{Main results}
\label{sec:main}

\subsection{Finite temperature}

\begin{theorem}[Finite temperature]
  \label{thm:finite-temperature}
  Suppose that $T > 0$. Then, the force constants decay exponentially: there exists
  $C,\eta > 0$ such that
  \[
    \big| C_{Rs,R's'} \big| \leq C e^{-\eta | R-R'|},
  \]
  for all $R, R' \in \mathbb L$ and $s,s'=1, \dots, N_{\rm at}$.
\end{theorem}

\begin{remark}
  By tracing the constant in the proof, one obtains $\eta \ge c T$ (for
  some constant $c > 0$ independent of $T$).
\end{remark}

  \begin{remark}
    For metals at zero temperature, one expects a slower decay due to
    Friedel oscillations. For instance, for the free electron gas, $W$
    is a multiplier in Fourier space given by
    \begin{align*}
      W(q)=\frac{4\pi}{|q|^{2} - 4\pi \chi_{0}(q)}
    \end{align*}
    where $\chi_{0}$ is a known function (the Lindhard response function)
    \cite{GiulianiVignale2005}. This function is radial, negative and has a
    finite value at $0$ (full screening), but has a logarithmic
    singularity in the derivative at $|q|=2k_{\rm F}$, where $k_{\rm F}$ is the Fermi wavevector.
    This singularity drives the decay of the IFC, which can be
    computed to be algebraic $|R-R'|^{-3}$ with oscillatory tails. 
    More
    generally, a similar result is expected to hold for metals with a
    well-defined Fermi surface \cite{Simion2005,Giuliani2005}. For
    other systems with a degenerate Fermi surface, explicit
    calculations are also possible (for instance, for graphene, one
    expects incomplete screening, with $\chi_{0}$ behaving like a
    constant at $q=0$).

    In both cases, the main difficulty is not in studying the
    asymptotics of $C_{Rs,R's'}$, but in showing the well-posedness of
    the defect problem and the regularity of $d \mapsto F^{d}$. A more
    careful study will be the topic of further research.
\end{remark}

\subsection{Insulators at zero temperature}
\begin{theorem}[Insulators]
  \label{thm:insulators}
  Suppose that $T=0$ and Assumption~\ref{assumption_gap} is satisfied. Then, the force constants decay algebraically:
  \begin{align}
    C_{Rs,R's'} &= - (Z_{s}^{\star})^{\t}\nabla^{2} 
    \Phi_{\rm M}(R-R'+\tau_{s}-\tau_{s'})
    Z_{s'}^{\star} + \mathcal O(|R-R'|^{-4})
    \nonumber
  \end{align}
  for all $R, R' \in \mathbb L$ and $s,s'=1, \dots, N_{\rm at}$, where
  \begin{align*}
    \Phi_{\rm M}(x) \coloneqq \frac 1 {\sqrt{\det \epsilon_{\rm M}}} \frac 1 {\sqrt{x^\t  \epsilon_{\rm M}^{-1} x}}
  \end{align*}
  is the
  {dielectric-constant screened Coulomb interaction} (the Green function
  of the Poisson equation $-{\rm div}(\epsilon_{\rm M} \nabla V) = 4\pi \rho$).
  Here, the {macroscopic
  dielectric constant} $\epsilon_{\rm M}$ is a positive definite $3 \times 3$ matrix, and the \emph{Born effective charges} $Z_{s}^{\star} \in \R^{3\times3}$ are
  defined in the proof.
\end{theorem}

\begin{remark}
  The interaction is of dipole-dipole type: it corresponds to the
  Coulomb interaction of a pair of dipoles at $R+\tau_{s}$ and
  $R'+\tau_{s'}$ oriented in the $\alpha$ and $\beta$ direction
  respectively. It is instructive to compare this result with the one
  obtained with simpler models.

  \begin{itemize}
  \item When ignoring the reaction of electrons (setting
    $\chi_{0} = 0$ so that $W = v_{\rm c}$), we get the same dipole-dipole
    interaction with $\epsilon_{\rm M}=1$ and $Z_{s}^{\star} = Z_{s} I$.
  \item In the opposite regime, we can imagine electrons as tightly bound
    to nuclei (although not necessarily their nuclei of origin),
    forming a cloud that moves rigidly with the atoms. The resulting
    system is composed of effective ions, with a partial charge equal
    to the number of protons of the atom minus the number of electrons
    bound to it. The resulting interaction is a dipole-dipole
    interaction with $Z_{s}^{\star}$ equal to these ionic charges, but
    $\epsilon_{\rm M}=1$. For instance, we may conceptualize the NaCl crystal as made
    up of ions Na$^{+}$ (charge $+1$) and Cl$^{-}$ (charge $-1$), and
    indeed, the Born effective charges calculated from density
    functional theory in simple
    ionic crystals are very close to these values.
  \end{itemize}
  The notion of Born effective charges can then be seen as a way to
  merge the insights from both previous models (ions with partial
  charges, and homogenized model leading to a dielectric constant).

  Note that when taking for the atomic density $m$ a diffuse function
  (i.e. scaling $m_{\lambda}(x) = \lambda^{-3} m(x/\lambda)$ with
  $\lambda$ large), $Z_{s}^{\star}$ converges to $Z_{s}I$, consistent
  with the scaling of \cite{Cances2010}. This regime corresponds to
  ignoring the lattice-scale oscillations (``local field effects''),
  in a way similar to homogenization scalings.
  \end{remark}

  \begin{remark}
    Our method of proof actually gives a full asymptotic expansion of
    $C_{Rs,R's'}$ in inverse powers of $x=R-R'+\tau_{s}-\tau_{s'}$.
    This expansion includes the usual multipole expansion (derivatives
    of $\Phi_{\rm M}$), but also other terms, derivatives of
    $|\epsilon_{\rm M}^{-1/2}x|^{2n-3}$ for $n \ge 0$. These additional terms can be
    rationalized as arising from a wavelength dependence of the
    dielectric constant.
  \end{remark}

\begin{remark}
  \label{rem:vanishing}
  The Born effective charges satisfy the sum rule
  $\sum_{s=1}^{N_{\rm at}} Z_{s}^{\star}=0$ (formally, this results
  from global translation invariance, but it is not easy to show it in
  our formalism). Therefore, in some systems, such as simple crystals with
  $N_{\rm at}=1$, the Born effective charges vanish. In this case, the
  leading-order asymptotics for the force constants is given by
  quadrupole-quadrupole interaction: there exist tensors
  $Z^{\star,2}_{s} = [Z^{\star,2}_{s,\alpha_1\alpha_2\beta}]_{\alpha_1\alpha_2, \beta} \in \mathbb R^{3^2 \times 3}$
  for which
   \begin{align}
     C_{Rs,R's'}
    &= -
    (Z_{s}^{\star,2})^\t 
    \,\,
    \nabla^{4} \Phi_{\rm M}( R - R' + \tau_{s} - \tau_{s'})
    \,\,
    Z^{\star,2}_{s'}
    + \mathcal O(|R-R'|^{-6}). \nonumber
\end{align}
More generally, faster decay is \textit{a priori} possible.
  \end{remark}

The proofs of Theorems~\ref{thm:finite-temperature} and \ref{thm:insulators} are contained in the following sections. 
\section{Translation Invariance}
\label{sec:translation}
In this section, we examine the consequences of translation
invariance, culminating in a reformulation of $C_{Rs,R's'}$ in
reciprocal space.
\subsection{Fourier transforms and series}

We use the following conventions for the three varieties of Fourier
transforms we will use, acting respectively on functions of $\R^{3}$,
$\mathbb L$-periodic functions (such as periodic parts of Bloch
waves), and sequences on $\mathbb L$ (such as displacements and
forces).

\begin{itemize}
\item For $\psi \in \mathcal S(\mathbb R^3)$, we define the Fourier of $\psi$ by
$$
\widehat \psi(k) := \frac{1}{(2\pi)^{3/2}} \int_{\mathbb R^3} \psi(x) e^{-ik \cdot x} \,
\D x. 
$$
The inverse Fourier transform is then given by
$$
\psi(x) = \frac{1}{(2\pi)^{3/2}} \int_{\mathbb R^3} \widehat \psi(k) e^{ik \cdot x} \, \D k.
$$
The Fourier transform and its inverse extend to unitary operators on $L^2(\mathbb R^3)$. We will also write $\widecheck{\psi}(x)$ for the inverse Fourier transform of $\psi$.

\item For $u \in C^\infty_{\rm per}(\Omega)$, we have
  \begin{align*}
    u(x) = \sum_{G \in \mathbb L^\star} c_{G}(u) e_{G}(x),\quad \text{with} \quad c_{G}(u) = \langle  e_{G}, u \rangle = \int_\Omega u(x) \frac{e^{-iG \cdot x}}{\sqrt{|\Omega|}} \, \D x,
  \end{align*}
  where $(e_G)_{G \in \L^\star}$ is the canonical Fourier basis of $L^2_{\rm per}\coloneqq L^2_{\rm per}(\Omega)$ (that is, $e_{G}(x) \coloneqq e^{iG \cdot x}/\sqrt{|\Omega|}$).
  The mapping $u \mapsto (c_G(u))_{G \in \mathbb L^\star}$ extends to a unitary operator from
  $L^2_{\rm per}$ to $\ell^{2}({\mathbb L}^\star)$.

\item For $d \in \ell^{\infty}(\mathbb L)$ with compact support, we have
  \begin{align*}
    d(R) = \frac
    {1}
    {\sqrt{|\Omega^\star|}}
    \int_{\Omega^\star} f(q) e^{-\i q \cdot  R}\D{q} , \quad \text{with} \quad  f(q) &= \frac
    {1}
    {\sqrt{|\Omega^\star|}} \sum_{R \in {\mathbb L}}
    d(R) e^{iq\cdot R}.
  \end{align*}
  The mapping $d \mapsto f$ extends to a unitary operator from
  $\ell^{2}(\mathbb L)$ to $L^{2}_{\rm per}(\Omega^{\star})$.
\end{itemize}

\subsection{Bloch transform}

The Bloch transform, $\mathcal B$, is unitary from $L^{2}(\mathbb{R}^{3})$ to
\begin{align*}
  L^{2}_{\rm qp}(\R^{3}, L^{2}_{\rm per}) \coloneqq 
  \left\{ 
    u_\bullet \in L^2_{\rm loc}(\mathbb R^3;L^2_{\rm per}) 
    \; | \; 
    u_{k+G}(x) = e^{-iG\cdot x}u_{k}(x)
    \,\, \forall \, G \in \mathbb L^\star 
    \mbox{, a.a. } x,k \in \mathbb R^3 \right\}
\end{align*}
with natural inner product 
$
    (u_\bullet,v_\bullet)_{L^{2}_{\rm qp}(\R^{3}, L^{2}_{\rm per})} 
    \coloneqq \int_{\Omega^\star} \braket{ u_k, v_k }_{L^2_{\mathrm{per}}} \D{k}.
$
For smooth functions $\psi$, we have
\begin{align*}
  (\mathcal B \psi)_{k}(x) &= u_k(x) = \frac 1 {\sqrt{|\Omega^{\star}|}} \sum_{R \in \mathbb L} \psi(x+R) e^{-ik \cdot (x+R)}, \\
  (\mathcal B^{-1} u_\bullet)(x) &= \psi(x) = \frac 1 {\sqrt{|\Omega^{\star}|}} \int_{\Omega^\star} u_k(x) e^{i k \cdot x} \D{k.}
\end{align*}
  We will use the simple relationship between the Bloch transform and the Fourier
  transforms/series
  \begin{align}
    \label{eq:fourier_bloch}
    \langle  e_{G}, u_{k} \rangle_{L^2_{\rm per}} &= \widehat \psi(k+G) \quad \mbox{for} \quad u_{\bullet} = \mathcal B \psi
  \end{align}
  (recall that $|\Omega||\Omega^{\star}|=(2\pi)^{3}$). Thanks to this
  relationship it is easy to see that the locality properties of $\psi$
  can be understood from the regularity of $u$, and vice-versa.

  \subsection{Periodic operators}
  
  The operator $W$ appearing in the key formula
\eqref{eq:CIJ_ij} for $C_{Rs,R's'}$ is $\L$-periodic.

As is standard, any bounded linear operator $A$ on $L^2(\mathbb R^3)$ commuting with lattice translations can be decomposed by the Bloch transform in the sense that there exists a function
$A_\bullet \in L^\infty_{\rm qp}(\mathbb R^3;\mathcal L(L^2_{\rm per}))$ such that for all
$u \in L^2(\mathbb R^3)$,
\begin{align*}
  (\mathcal B (A u))_{k} = A_{k} (\mathcal B u)_{k}.
\end{align*}
Here, the
subscript $\rm qp$ means that $A_{k + G} = e^{-iG \cdot \bullet} A_k e^{iG \cdot \bullet}$ for all $G \in \mathbb L^\star$ and a.a.~$k\in \mathbb R^3$.

Similar properties hold for $\mathbb L$-translation invariant
unbounded operators on $L^2(\mathbb R^3)$ or bounded operators from
$\mathcal X$ to $\mathcal Y$ (this is the case for $W$), where $\mathcal X$ and $\mathcal Y$ are
Hilbert subspaces of ${\mathcal D}'(\mathbb R^3)$ with appropriate
modifications accounting for domains.

It is sometimes convenient to introduce the Bloch matrices
$[A]_{GG'}(k)$ for $G,G' \in \mathbb L^\star$, representing the operators $A_k$
in the basis $(e_{G})_{G \in{\mathbb L}^\star}$:
\begin{align}
  \label{eq:bloch_matrix}
  [A]_{GG'}(k):= \langle e_G, A_k e_{G'} \rangle_{L^2_{\rm per}}.
\end{align}
These matrices have the properties that
\begin{align}
  \label{eq:bloch_matrices}
  \widehat{(A\psi)}(k+G) = \sum_{G' \in \mathbb L^\star} [A]_{GG'}(k) \; \widehat \psi(k+G')
\end{align}
for all $G \in \mathbb L^\star \mbox{ and a.a. } k \in \mathbb R^3$.

\subsection{The dynamical matrix}

The force constants are invariant with respect to lattice
translations, $C_{Rs,R's'} = C_{(R+T)s,(R'+T)s'}$ for all $T \in \mathbb L$, and thus the linearized force
\begin{align*}
  F_{Rs} = \sum_{R',s'} C_{Rs,R's'} d_{R's'}
  \eqqcolon 
  \sum_{s'} \sum_{R'} C_{ss'}(R-R') d_{R's'}
  ,
\end{align*}
is given by convolutions with kernel
$C_{ss'}(R) \coloneqq C_{Rs,0s'}$. By Theorem~\ref{thm:W}, the sequence $C_{ss'}(R)$ is bounded and
therefore admits a distributional Fourier transform, the \textit{dynamical
  matrix}, given by (in a weak
sense, i.e.~tested against a ${\mathbb L}^\star$-periodic and smooth test function)
\begin{align}
  \label{eq:Dst-1}
  D_{ss'}(q) &= \frac  {1} {\sqrt{|\Omega^\star|}} \sum_{R \in {\mathbb L}} e^{iq\cdot R} C_{ss'}(R).
\end{align}
By Fourier duality, questions about the locality of $C_{ss'}$ can be formulated in terms
of the smoothness of $D_{ss'}$.

\begin{remark}[Phonons]
  The dynamical matrix encodes the free oscillations of the atoms of
  the crystal in the linearized approximation. In particular, the
  eigenvalues of an appropriately mass-scaled version of $D(q)$ give
  access to the phonon band diagram.
\end{remark}

Combining the definition \eqref{eq:Dst-1} of $D_{ss'}$ with the
expression \eqref{eq:CIJ_ij} of $C_{ss'}(R)$, and using the Bloch
expression of $W$, we obtain

\begin{align*}
  &\big[D_{ss'}(q)\big]_{\alpha\beta} 
  = \frac  {1} {\sqrt{|\Omega^\star|}} \left( \delta_{ss'} [c_{s}]_{\alpha\beta} + Z_{s}Z_{s'} \bigg\langle  \sum_{R \in \mathbb L} e^{-iq\cdot R} \partial_{\alpha{}} m (\cdot - R- \tau_s),   
    W \!\left( \partial_{\beta} m(\cdot -\tau_{s'}) \!\right)  \!\!\bigg\rangle_{\!\! L^2} \right)\\
  &=
  \frac  {1} {\sqrt{|\Omega^\star|}} \Bigg( \delta_{ss'} [c_{s}]_{\alpha\beta} \nonumber\\
  &\qquad + Z_{s}Z_{s'} \int_{\Omega^{\star}} \sum_{\substack{R \in \mathbb L\\G, G' \in {\mathbb L}^\star}} e^{-i(q-k)\cdot R} \overline{\widehat{\partial_{\alpha{}} m (\cdot - \tau_s)}(k+G)} [W]_{GG'}(k) \widehat{\partial_{\beta} m(\cdot-\tau_{s'})}(k+G') \D k \Bigg)
  \\
  &= \! 
  \frac
  {\delta_{ss'}}
  {\sqrt{|\Omega^\star|} }
  [c_s]_{\alpha\beta} 
  \!+ \sqrt{|\Omega^\star|} Z_s Z_{s'}
  \!\!\!\!\!
  \sum_{G,G' \in \mathbb L^\star} 
  \!\!\!\! \overline{\widehat{\partial_\alpha m(\cdot\!-\!\tau_s)}(q \! + \! G)} 
  \! \left[ W \right]_{GG'}\!(q) 
  \widehat{\partial_\beta m(\cdot\!-\!\tau_{s'})}(q \! + \! G') 
\end{align*}
again in a weak sense, and where we have used
\eqref{eq:bloch_matrices} to go from the first line to the second and
the formula
$\sum_{R \in \mathbb L} \int_{\Omega^{\star}}e^{ik\cdot R} f(k)\D k = |\Omega^{\star}|f(0)$
to go from the second to the third.

\subsection{Strategy of proof}

In order to conclude the proof of Theorems~\ref{thm:finite-temperature} and \ref{thm:insulators}, we show that \textit{(i)} at finite temperature $D_{ss'}$ is an $\L^\star$-periodic analytic function, and \textit{(ii)} for insulators at zero temperature, the $\L^\star$-periodic function $D_{ss'}$ is analytic away from $\L^*$ and we are able to compute the inverse Fourier transform explicitly near $q = 0$.

Recall that $W = \epsilon^{-1} v_{\rm c}$. All these operators are periodic, with fibers given by
\begin{align*}
  [v_{{\rm c}}]_{GG'}(q) = \delta_{GG'} \frac {4\pi} {|q+G|^{2}}, \quad 
  \epsilon_{q} = 1 - v_{{\rm c},q} \chi_{0,q}, \quad
  [W]_{GG'}(q) = [\epsilon^{-1}]_{GG'}(q) [v_{\rm c}]_{GG'}(q).
\end{align*}
We therefore need to understand the operator $\chi_{0,q}$.

The main obstacle to locality is the long-range Coulomb interaction,
manifested by the divergence of $[v_{\rm c}]_{00}(q)$ as $q \to 0$. The
drastic difference between the zero-temperature insulator and the finite temperature case
comes down to the different behavior of $[\chi_{0}]_{00}$ for small $q$:
\begin{align}  \label{eq:DOS}
  \lim_{q \to 0} \,\, [\chi_0]_{00}(q) 
  &= -\mathrm{DOS}
  \qquad \text{where} \qquad \mathrm{DOS} := -\sum_{n=1}^{+\infty} \fint_{\Omega^\star} f_{T}'(\ep_{nk}-\mu^0) \D{k} \ge 0
\end{align}
is the density of states at the Fermi level, positive for finite
temperature and zero for insulators at zero temperature. 

At finite temperature, the density of states at the Fermi level is
non-zero, and so $ [\chi_0]_{00}(q)$ behaves as a constant, whereas,
for insulators at zero temperature, ${\rm DOS}=0$ and the leading term
is quadratic. As a result, at finite temperature,
$[\epsilon^{-1}]_{00}(q) \sim_{q \to 0} \frac{|q|^2}{\rm DOS}$, which
smoothes out the singularity of the term $[v_{\rm c}]_{00}(q)=\frac{4\pi}{|q|^2}$.
In this case, we observe an exponential decay of the interatomic force
constants. On the other hand, for insulators at zero temperature,
$[\epsilon^{-1}]_{00}(|q| e)$ goes to a constant $1/(e^{\t} \epsilon_{\rm M} e)$ as
$|q| \to 0$ (partial screening). This leads to an algebraic decay of
the force constants. In the following, we will rigorously prove this
heuristic description, starting with the finite temperature case.

\section{Finite Temperature}
\label{sec:pf-metals}

Since $\dos > 0$, it can be shown that $-\chi_{0,q} + v_{{\rm c},q}^{-1}$ is
self-adjoint and positive with
$W_q = ( -\chi_{0,q} + v_{{\rm c},q}^{-1})^{-1} \colon L^2_{\rm per} \to L^2_{\rm per}$
bounded linear operator \cite[Proof of Lemma~5.2]{Levitt2020}. The
same proof also shows that the mapping
$q \mapsto W_q \in \mathcal L(L^{2}_{\rm per}, L^{2}_{\rm per})$ is analytic on a strip $\R^{3}+i[-a,a]^{3}$ for some $a > 0$. 

Let $w_{a}(x) = e^{a \sqrt{1+|\cdot|^{2}}}$ be the exponential weight with
exponent $a > 0$. It follows
from the relationship between Bloch matrices and Fourier transforms
\eqref{eq:bloch_matrices} and usual Paley--Wiener arguments that there
is $a > 0$ such that $w_{-a} W w_{a}$ is bounded on $L^{2}$.

By translation invariance, it is enough to consider $C_{Rs,R's'}$ for
$R'=0$. Then, for $R\not=0$, we have
\begin{align*}
  [C_{Rs,0s'}]_{\alpha\beta} &=Z_s Z_{s'} \bigg\langle   \frac{\partial m}{\partial x_\alpha} (\cdot - R - \tau_s), 
    W   \frac{\partial m}{\partial x_\beta}(\cdot - \tau_{s'})   \bigg\rangle_{\!\! L^2}\\
  &= Z_s Z_{s'}
    \bigg\langle  w_{a} \frac{\partial m}{\partial x_\alpha} (\cdot - R - \tau_s), 
    (w_{-a} W w_{a}) w_{-a} \frac{\partial m}{\partial x_\beta}(\cdot - \tau_{s'})   \bigg\rangle_{\!\! L^2},
    \end{align*}
    so that
    \begin{align*}
    \left|[C_{Rs,0s'}]_{\alpha\beta}\right| &\lesssim 
    \left\|w_{a} \frac{\partial m}{\partial x_\alpha} (\cdot - R - \tau_s)\right\|_{L^2}
    \left\|w_{-a} \frac{\partial m}{\partial x_\beta}(\cdot - \tau_{s'})\right\|_{L^2}
    \lesssim e^{-a'|R|}
  \end{align*}
  for some $a' < a$.

This concludes the proof of Theorem~\ref{thm:finite-temperature}.

\section{Insulators at Zero Temperature}
\label{sec:locality}

We now consider the case $T = 0$ and suppose Assumption~\ref{assumption_gap} is satisfied.

\subsection{Step 1. Sum-over-states formulas}

Explicit expressions of the Bloch matrices of $\chi_0$ and $W$ can be obtained from a spectral decomposition of the Bloch fibers 
$$
H^0_{k} = -\frac 12 \Delta_k + V^0_{\rm per} = \frac12 (-i\nabla + k)^2 + V^0_{\rm per}
$$
of 
$H^{0}_{\rm per} = -\tfrac 12 \Delta + V^0_{\rm per}$. The $H^0_{k}$'s are self-adjoint operators on $L^2_{\rm per}$ with common domain $H^2_{\rm per}$. For all $k \in \mathbb R^3$, we consider an orthonormal basis $(u_{n,k})_{n \in \mathbb N^\star}$ of $L^2_{\rm per}$ consisting of eigenfunctions of $H^0_k$ associated with the eigenvalues of this operator, ranked in non-decreasing order:
$$
H^0_{k} u_{n,k} = \varepsilon_{n,k} u_{n,k}, \quad \langle u_{m,k},u_{n,k} \rangle_{L^2_{\rm per}} = \delta_{mn}, \quad  \varepsilon_{1,k} \le \varepsilon_{2,k} \le \cdots.
$$
Assumption~\ref{assumption_gap} states that $\max\limits_{k \in \Omega^\star} \ep_{N_{\rm el},k} < \mu^0 < \min\limits_{k \in \Omega^\star}\ep_{N_{\rm el}+1,k}$.

Let $\mathscr C$ be a simple closed positively oriented contour in $\{ z \in \mathbb C \colon \re z < \mu^0\}$ encircling $\bigcup_{k\in \Omega^\star}\{ \ep_{n,k} \}_{n=1}^{ N_{\rm el} }$ (see Figure~\ref{fig:contour}) and write 
$$
    f_0(H_{\rm per}^0 + V - \mu^0) = \oint_{\mathscr C} (z - H_{\rm per}^0 + V)^{-1} \frac{\D{z}}{2\pi i},
$$
and thus
$$
\chi_0 V = {\rm den} \oint_{\mathscr C} (z - H_{\rm per}^0)^{-1} V (z - H_{\rm per}^0)^{-1} \frac{\D{z}}{2\pi i}
$$
(recall $\chi_0$ is the derivative of $V \mapsto {\rm den}\left(f_0(H_{\rm per}^0 + V - \mu^0))\right)$ at $V = 0$).

\begin{figure}[ht]
    \centering

\tikzset{every picture/.style={line width=0.75pt}} 

\begin{tikzpicture}[x=0.75pt,y=0.75pt,yscale=-1,xscale=1]

\draw  [fill={rgb, 255:red, 0; green, 0; blue, 0 }  ,fill opacity=1 ][line width=2.25]  (204,67.5) -- (293,67.5) -- (293,73) -- (204,73) -- cycle ;
\draw  [color={rgb, 255:red, 0; green, 0; blue, 0 }  ,draw opacity=1 ][line width=2.25]  (400,68.5) -- (566,68.5) -- (566,76.5) -- (400,76.5) -- cycle ;
\draw  [color={rgb, 255:red, 21; green, 0; blue, 255 }  ,draw opacity=1 ][line width=2.25]  (48,38.6) .. controls (48,26.12) and (58.12,16) .. (70.6,16) -- (323.4,16) .. controls (335.88,16) and (346,26.12) .. (346,38.6) -- (346,106.4) .. controls (346,118.88) and (335.88,129) .. (323.4,129) -- (70.6,129) .. controls (58.12,129) and (48,118.88) .. (48,106.4) -- cycle ;
\draw [shift={(185,16)}, rotate = 359.67] [fill={rgb, 255:red, 0; green, 0; blue, 255 }  ,fill opacity=1 ][line width=0.08]  [draw opacity=0] (23.22,-11.15) -- (0,0) -- (23.22,11.15) -- cycle    ;
\draw [shift={(210,128.8)}, rotate = 179.67] [fill={rgb, 255:red, 0; green, 0; blue, 255 }  ,fill opacity=1 ][line width=0.08]  [draw opacity=0] (23.22,-11.15) -- (0,0) -- (23.22,11.15) -- cycle    ;
\draw  [fill={rgb, 255:red, 0; green, 0; blue, 0 }  ,fill opacity=1 ] (610,72) .. controls (610,69.79) and (611.79,68) .. (614,68) .. controls (616.21,68) and (618,69.79) .. (618,72) .. controls (618,74.21) and (616.21,76) .. (614,76) .. controls (611.79,76) and (610,74.21) .. (610,72) -- cycle ;
\draw  [fill={rgb, 255:red, 0; green, 0; blue, 0 }  ,fill opacity=1 ] (624,72) .. controls (624,69.79) and (625.79,68) .. (628,68) .. controls (630.21,68) and (632,69.79) .. (632,72) .. controls (632,74.21) and (630.21,76) .. (628,76) .. controls (625.79,76) and (624,74.21) .. (624,72) -- cycle ;
\draw  [fill={rgb, 255:red, 0; green, 0; blue, 0 }  ,fill opacity=1 ] (637,72) .. controls (637,69.79) and (638.79,68) .. (641,68) .. controls (643.21,68) and (645,69.79) .. (645,72) .. controls (645,74.21) and (643.21,76) .. (641,76) .. controls (638.79,76) and (637,74.21) .. (637,72) -- cycle ;
\draw  [fill={rgb, 255:red, 0; green, 0; blue, 0 }  ,fill opacity=1 ][line width=2.25]  (102,67.5) -- (178,67.5) -- (178,73.5) -- (102,73.5) -- cycle ;

\draw (215,77.4) node [anchor=north west][inner sep=0.75pt]    {$\varepsilon _{N_{\mathrm{el}}}\left( \Omega ^{\star }\right)$};
\draw (443,77.4) node [anchor=north west][inner sep=0.75pt]    {$\varepsilon _{N_{\mathrm{el}} +1}\left( \Omega ^{\star }\right)$};
\draw (311,19.4) node [anchor=north west][inner sep=0.75pt]  [font=\Large,color={rgb, 255:red, 21; green, 0; blue, 248 }  ,opacity=1 ]  {$\mathscr{C}$};

\end{tikzpicture}

    \caption{Contour encircling the occupied spectrum}
    \label{fig:contour}
  \end{figure}
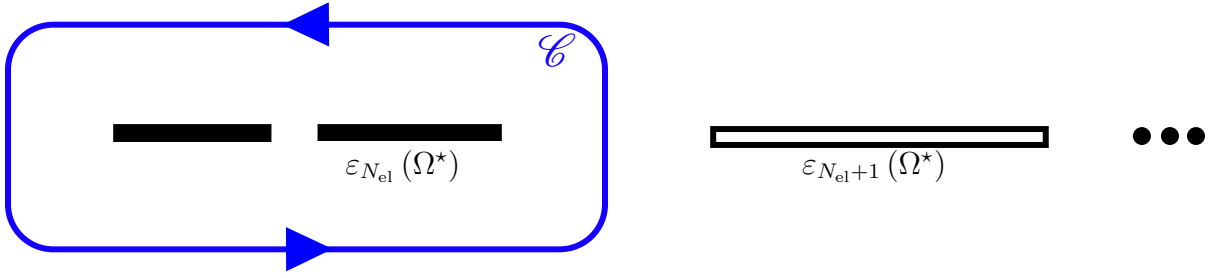

Following the computations of \cite[Lemma 5.1]{Levitt2020}, we have, with
the notation $\fint_{\Omega^{\star}} = \tfrac 1 {|\Omega^{\star}|}\int_{\Omega^{\star}}$,
\begin{align} \label{eq:chi0KKq-1}
  [\chi_0]_{GG'}(q) 
  &= \oint_{\mathscr C} 
  \fint_{\Omega^\star}
  {\Tr}_{L^{2}_{\rm per}} \left[ 
    \overline{e_G}
    \big(z-H^0_{k + q}\big)^{-1}
    e_{G'}  \big(z - H^0_k\big)^{-1}
  \right] \D{k}
  \frac{\D{z}}{2\pi \i} 
  \nonumber\\
  &= \oint_{\mathscr C} 
    \fint_{\Omega^\star}
    \sum_{n,m=1}^\infty
    \frac
        {\braket{e_G u_{mk}, u_{n,k+q}}
      \braket{u_{n,k+q}, e_{G'} u_{mk}} }
        {(z-\ep_{n,k + q}) (z-\ep_{mk}) } 
    \D{k}
    \frac{\D{z}}{2\pi \i} \\
  &= \sum_{n,m=1}^{+\infty}
  \fint_{\Omega^\star} 
  \frac
  {f_{0}(\ep_{n,k+q}-\mu^0) - f_{0}(\ep_{mk}-\mu^0)}
  {\ep_{n,k + q} - \ep_{mk}}
  \Braket{
    e_{G} u_{mk}, u_{n,k + q}
  }
  \Braket{
    u_{n,k + q}, e_{G'} u_{mk}
  }
  \D{k}.\nonumber
\end{align} 
This last expression is known as the
Adler--Wiser sum-over-states formula \cite{Adler1962,Wiser1963}. One
may see that the summation in the second line converges absolutely by
applying Cauchy--Schwartz inequality combined with
\begin{align}
    \sum_{n,m=1}^\infty
    \frac
        {\left|\braket{e_G u_{mk}, u_{n,k+q}}\right|^2}
        {|z-\ep_{n,k+q}|^2} 
    = \sum_{n=1}^\infty
    \frac
        {\| e_{-G} u_{n,k+q} \|^2_{L^2_{\mathrm{per}}}}
        {|z-\ep_{n,k+q}|^2} 
    \leq \frac{1}{|\Omega|} \sum_{n=1}^\infty
    \frac
        {1}
        {|z-\ep_{n,k+q}|^2}
\end{align}
which is uniformly bounded for $z\in \mathscr C$ and $k, k+q\in \Omega^\star$. Here,
we have used the fact that $(u_{mk})_m$ is a orthonormal basis of
$L^2_{\mathrm{per}}$ and that there exists $a,b > 0$ such that
$\ep_{nk} \geq a n^{\frac{2}{3}} - b$ for all $n$ (this readily follows from the Weyl asymptotic formula for periodic Schr\"odinger operators on $\R^3$). Therefore, one is able to integrate over the contour term-by-term to obtain the final line by Cauchy's integral formula.

\subsection{Step 2. The screened Coulomb operator} Using a Schur complement, we
can study the behavior of $[W]_{GG'}(q)$ for $q$ near $0$:

\begin{lemma}[Asymptotic expansion of $W_q$]
  The function $q\mapsto W_q$ is analytic from $\R^{3} \setminus {\mathbb L}^\star$ to
  $\mathcal L(L^{2}_{\rm per}, L^{2}_{\rm per})$. Moreover, there is a neighborhood $U$ of $0$ such that, for all $q \in U$ and $G,G'\in \mathbb L^\star$, we have
\label{lem:W}
\begin{align}
[W]_{GG'}(q) &= 
\frac
{4\pi \, w_{G}(q) \overline{w_{G'}(q)}}
    {q^{\t} \epsilon_{\rm M} q + R(q)}
\quad \text{with} \quad 
w_G(q) \coloneqq 
\begin{cases}
    1 &\text{if } G=0, \\
    b_G \cdot q + r_{G}(q) &\text{if } G\not=0,
\end{cases}
\label{eq:WGG'}
\end{align}
$b_G \in \mathbb C^3$, and $R,r_G: U \to \mathbb C$ are analytic with $\partial^\alpha R(0) = 0$ for all $|\alpha|\leq 2$ and $\partial^\beta r_G(0) = 0$ for all $|\beta|\leq 1$ and all $G \in \mathbb L^\star$.
\end{lemma}
\begin{proof}
  Here and in the rest of this proof, we denote by $\mathcal O_{a}(|q|^{n})$ a
  quantity that is analytic with respect to $q$ in the appropriate
  topologies ($\R$, $L^{2}_{\rm per}$ or
  $\mathcal L(L^{2}_{\rm per}, L^{2}_{\rm per})$), and has vanishing
  $n-1$ first derivatives at $q=0$.

One can show that $(z-H_k^0)^{-1} (1 - \Delta_k)$ and its inverse are bounded
uniformly for $z \in \mathscr C$ and $k \in \Omega^\star$
\cite[Lemma~3]{Cances2008}, and thus
$\| (z - H_k^0)^{-1} (H_{k+q}^0 - H_k^0) \| \lesssim |q|$ uniformly in $z\in\mathscr C$ and
$k,q\in \Omega^\star$ with $|q|<1$. As a result, $q \mapsto (z - H_{k+q}^0)^{-1}$ is given for small $q$
by the Neumann series
$\sum_{n=0}^\infty \big( (z - H_k^0)^{-1} (H_{k+q}^0 - H_k^0) \big)^n (z - H_k^0)^{-1} $
and therefore $q\mapsto \chi_{0,q}$ is analytic from
$\Omega^\star$ to $\mathcal L(L^2_{\rm per}, L^2_{\rm per})$.

Using the orthogonality of $(u_{nk})_{n}$ in \cref{eq:chi0KKq-1}, we
have $[\chi_0]_{0G}(0) = [\chi_0]_{G0}(0) = 0$ for all $G \in {\mathbb L}^\star$.
Moreover,
\begin{align*}
  \nabla[\chi_0]_{00}(q=0) &= \oint_{\mathscr C} 
  \fint_{\Omega^\star}
  {\Tr}_{L^{2}_{\rm per}} \left[ \overline{e_{0}}
    \big(z-H^0_{k}\big)^{-1}
    \nabla_k H^{0}_{k}
    \big(z-H^0_{k}\big)^{-1}
    {e_{0}}
    \big(z - H^0_k\big)^{-1}
  \right] \D{k}
  \frac{\D{z}}{2\pi \i}\\
  &= \frac 1 {|\Omega|}\oint_{\mathscr C} 
  \fint_{\Omega^\star}
  {\Tr}_{L^{2}_{\rm per}} \left[ 
    \nabla_k H^{0}_{k}
    \big(z - H^0_k\big)^{-3}
  \right] \D{k}
  \frac{\D{z}}{2\pi \i}
  = 0
\end{align*}
because the integrand has only poles of order $3$, which have no residues.
It then follows that, in the orthogonal decomposition 
    $L^2_{\rm per} = \mathbb C e_0 \oplus e_0^\perp$, 
 we have
\begin{align*}
  \chi_{0,q} &=
  \begin{pmatrix}
    - \frac{1}{4\pi} q^{\t} L q + \mathcal O_a(|q|^{3}) & B^{\star} \cdot q + \mathcal O_a(|q|^{2})\\
    B \cdot q + \mathcal O_a(|q|^{2}) & \chi_{0,q}^{\neq 0}
  \end{pmatrix}
\end{align*}
where $q\mapsto \chi_{0,q}^{\neq 0}$ is analytic from $\R^{3}$ to $\mathcal L(e_{0}^{\perp}, e_{0}^{\perp})$ and $L \in \mathbb R^{3\times3}_{\mathrm{sym}}, B_G \in \mathbb C^3$ are given by
\begin{align}
    &L_{\alpha\beta} = 2i \oint_{\mathscr C}
  \fint_{\Omega^\star} \!\!
  {\Tr}_{L^{2}_{\rm per}} \!\! \left[ 
    \overline{e_0}
    \big(z-H^0_{k}\big)^{-1}
    \partial_{k_\alpha} H_{k}^0
    \big(z-H^0_{k}\big)^{-1}
    \partial_{k_\beta} H_k^0 
    \big(z-H^0_{k}\big)^{-1}
    e_0 \big(z - H^0_k\big)^{-1}
  \right] \! \D{k}\D{z}, \nonumber\\
   &[B_G]_\alpha = \oint_{\mathscr C}
  \fint_{\Omega^\star}
  {\Tr}_{L^{2}_{\rm per}} \left[ 
    \overline{e_{G}}
    \big(z-H^0_{k}\big)^{-1}
    \partial_{k_\alpha} H_{k}^0
    \big(z-H^0_{k}\big)^{-1}
    e_0 \big(z - H^0_k\big)^{-1}
  \right] \D{k}\frac{\D{z}}{2\pi i}.
\end{align}

Using \cite[Lemma~6]{Cances2010}, $-\chi_{0,q}^{\neq 0} + (v_{{\rm c},q}^{\neq 0})^{-1}$ is bounded and invertible as an
operator between $\sqrt{v_{{\rm c},q}} P_{e_0^\perp} L^2_{\rm per}$ and $\frac{1}{\sqrt{v_{{\rm c},q}}} P_{e_0^\perp} L^2_{\rm per}$.
Therefore, it follows from a Schur complement argument that $W_{q} \colon \frac{1}{\sqrt{v_{{\rm c},q}}}L^2_{\rm per} \to \sqrt{v_{{\rm c},q}} L^2_{\rm per}$ may be decomposed as
\begin{align*}
  W_{q} &= (-\chi_{0,q} + v_{{\rm c},q}^{-1})^{-1}\\
  &=\begin{pmatrix}
    - \frac{1}{4\pi} q^{\t} (1+L) q + \mathcal O_a(|q|^{3}) & -B^{\star} \cdot q + \mathcal O_a(|q|^{2})\\
    -B \cdot q + \mathcal O_a(|q|^{2}) & -\chi_{0,q}^{\neq 0} + (v_{{\rm c},q}^{\neq 0})^{-1}
  \end{pmatrix}^{-1}\\
&=
  \frac
    {4\pi}
    {q^\t \epsilon_{\rm M} q + \mathcal O_a(|q|^3)}
  \begin{pmatrix}
    1 & b^\star \cdot q + \mathcal O_a(|q|^{2}) \\
    b \cdot q + \mathcal O_a(|q|^{2}) & (b \cdot q) \otimes (b^\star \cdot q) + \mathcal O_a(|q|^3)
  \end{pmatrix}
\end{align*}
with  
\begin{align}
  [ \epsilon_{\rm M} ]_{\alpha\beta} &= \delta_{\alpha\beta} + L_{\alpha\beta} - \sum_{G,G'\not=0} [B^\star_G]_\alpha \,\, \big[(-\chi_{0,q}^{\neq 0} + (v_{{\rm c},q}^{\neq 0})^{-1})^{-1}\big]_{GG'}(0) \,\,  [B_{G'}]_{\beta} 
    \qquad \text{and} \nonumber\\ 
    [b_G]_\alpha &= - \sum_{G'\not= 0} \big[(-\chi_{0,q}^{\neq 0} + (v_{{\rm c},q}^{\neq 0})^{-1})^{-1}\big]_{GG'}(0) \,\, [B_{G'}]_\alpha,
\end{align}
as required.\end{proof}

We now convert Lemma~\ref{lem:W} on the smoothness of $W_{q}$ to our
main results on locality.

\subsection{Step 3. Reduction to the singular component}
\label{sec:pf-insul}

Here as in the previous
section, we denote by $\mathcal O_{a}(|q|^{n})$ a quantity that is analytic
with respect to $q$ in a neighborhood of $0$, and has vanishing $(n-1)$ first derivatives at $q=0$.

Using Lemma \ref{lem:W}, we obtain that $D_{ss'}(q)$ is
analytic on $\R^{3} \setminus {\mathbb L}^\star$ and that, in a neighborhood of $q=0$,
\begin{align}
    &\big[D_{ss'}(q)\big]_{\alpha\beta}
    = \frac
        {\delta_{ss'}}
        {\sqrt{|\Omega^\star|}} 
    [c_s]_{\alpha\beta} 
    \! + \!\!
    \sqrt{|\Omega^\star|}
    \!\!\!\!
    \sum_{G,G' \in \mathbb L^\star} 
    \!\!\!\!\!\overline{Z_s\widehat{\partial_\alpha m(\cdot\!-\!\!\tau_s)}(q \! + \! G)} 
    \! \left[ W \right]_{GG'}\!(q) 
    Z_{s'}\widehat{\partial_\beta m(\cdot\!-\!\!\tau_{s'})}(q \! + \! G') \nonumber\\
    &= \frac
        {4\pi \sqrt{|\Omega^\star|} e^{i q \cdot ( \tau_{s} - \tau_{s'} )}}
        {q^\t \epsilon_{\rm M} q + \mathcal O_a(|q|^3)}
    \sum_{G\in \mathbb L^\star}\overline{ 
            \overline{w_G(q)}
            Z_s\widehat{\partial_\alpha m}(q + G)
            e^{-i G\cdot \tau_s}
        }
        \sum_{G'\in \mathbb L^\star} 
        \overline{w_{G'}(q)}Z_{s'}\widehat{\partial_\beta m}(q + G')
        e^{-i G'\cdot \tau_{s'}}
    \nonumber\\ &\hspace{10.5cm}+ O_{a}(1) \nonumber\\
    &= \frac{4\pi \sqrt{|\Omega^\star|}}{(2\pi)^3} e^{i q \cdot ( \tau_{s} - \tau_{s'} )}
    \underbrace{\frac
    {(\sum_{\gamma} q_\gamma \overline{Z_{s,\gamma\alpha}^\star})
      (\sum_{\delta} q_\delta Z_{s',\delta\beta}^\star) + \mathcal O_a(|q|^{3})}
    {q^\t \epsilon_{\rm M} q + \mathcal O_a(|q|^{3})}}_{f(q)}
      + O_{a}(1) \nonumber
\end{align}
where
\begin{align}
  Z_{s,\alpha\beta}^\star
    &= Z_s
    \left[ 
        \delta_{\alpha\beta}
        + 
        (2\pi)^{\frac{3}{2}} \sum_{G\in\mathbb L^\star\setminus\{0\}} 
        \, 
        \overline{[b_G]_\alpha} 
        \, 
        G_{\beta}
        \,
        \widehat{m}(G) e^{-iG\cdot \tau_s}
    \right].
    \label{eq:Z*}
  \end{align}
  By time-reversal symmetry, we have $(u_{n,-k}, \ep_{n,-k}) = (\overline{u_{nk}}, \ep_{nk})$, so that $[\chi_0]_{GG'}(q) = \overline{ [\chi_0]_{-G,-G'}(-q) }$ and thus $b_{G} = - \overline{b_{-G}}$ and $Z_{s}^\star \in \mathbb R^{3\times3}$.

  Introduce now a reciprocal-space smooth non-negative radial cut-off function $\eta$ which is equal to $1$ in
  a small neighborhood $U$ of $0$, and to $0$ outside $2U$. By
  choosing $\Omega^{\star}$ to contain $2U$, we get that
  $D(q) (1-\eta(q))$ is analytic on $\Omega^{\star}$ and extends to a
  periodic function. Therefore, the corresponding component of
  $C_{Rs,R's'}$ decays faster than any inverse polynomial, which we
  denote by $O(|R-R'|^{-\infty})$. It follows that
  \begin{align}
    C_{Rs,R's'}
    &= \frac 1 {\sqrt{|\Omega^{\star}|}}\int_{\Omega^{\star}} D_{ss'}(q) e^{iq\cdot (R-R')} \D{q} \nonumber\\
    &= \frac{4\pi}{(2\pi)^3} \int_{\R^{3}} e^{i q \cdot ( \tau_{s} - \tau_{s'} )} \eta(q) f(q) e^{iq\cdot (R-R')} \D{q} + O(|R-R'|^{-\infty}). \nonumber
  \label{eq:lead}
  \end{align}
Next, we study the first term in the right-hand side of the above relation.
  
  \subsection{Step 4: Explicit computation of the singular
    component.}\label{sec:step4} 
  By Taylor expanding the numerator of $f$ and by expanding its
  denominator in an incomplete geometric series, we get
  \begin{align*}
    f(q) &=\frac
    {(\sum_{\gamma} q_\gamma \overline{Z_{s,\gamma\alpha}^\star})
      (\sum_{\delta} q_\delta Z_{s',\delta\beta}^\star) + \mathcal O_{a}(|q|^{3})}
    {q^\t \epsilon_{\rm M} q + \mathcal O_{a}(|q|^{3})}
    \\
    &=\underbrace{\frac {(\sum_{\gamma} q_\gamma \overline{Z_{s,\gamma\alpha}^\star})
      (\sum_{\delta} q_\delta Z_{s',\delta\beta}^\star)}
    {q^\t \epsilon_{\rm M} q}}_{f_{\rm lead}(q)} + \underbrace{\sum_{n=1}^{N-1} \frac{p_{n}(q)}{(q^{\t}\epsilon_{\rm M}q)^{N}}}_{f_{\rm hom}(q)} + \underbrace{\frac
    {\mathcal O_a(|q|^{3N+2})}
    {(q^\t \epsilon_{\rm M}q)^{N+1}+\mathcal O_a(|q|^{2N+3})}}_{f_{\rm rem}(q)},
\end{align*}
where $p_n$ are homogenous polynomials of degree $n+2N$.

By fixing in the rest of this proof $N = 5$, $f_{\rm rem}$ is $C^{4}$, so that the
corresponding contribution in $C_{Rs,R's'}$ is $O(|R-R'|^{-4})$.
We then treat the first two terms in turn:

\textit{(i) Leading order contribution.} Using the general fact that
    $\widehat{f\circ A} = \left|\det A\right|^{-1} \widehat{f} \circ A^{-\t}$ for
    invertible $A\in \mathbb R^{3\times3}$, we obtain
    \begin{align}
      \left[ 
        \big|\sqrt{\epsilon_{\rm M}} q\big|^{-2}
      \right]^\vee(x)
      = \frac
        { \widecheck{|\cdot|^{-2}}( \epsilon_{\rm M}^{-\frac{1}{2}} x) }
        {\sqrt{\det \epsilon_{\rm M}}} 
      &= \frac{(2\pi)^{\frac32}}{4\pi} 
      \frac
        {|x^{\t} \epsilon_{\rm M}^{-1} x|^{-\frac{1}{2}}}
        {\sqrt{\det \epsilon_{\rm M}}} 
       = 
       \frac{(2\pi)^{\frac32}}{4\pi} 
        \Phi_{\rm M}(x). \nonumber
    \end{align}
    Therefore, for the corresponding contribution in $C_{Rs,R's'}$ we obtain
    \begin{align*}
      \frac{4\pi}{(2\pi)^3} &\int_{\R^{3}} e^{i q \cdot ( \tau_{s} - \tau_{s'} )} \eta(q) f_{\rm lead}(q) e^{iq\cdot (R-R')} \D{q} \nonumber\\
      &=       \frac{4\pi}{(2\pi)^3} \int_{\R^{3}} e^{i q \cdot ( \tau_{s} - \tau_{s'} )} f_{\rm lead}(q) e^{iq\cdot (R-R')} \D{q} + O(|R-R'|^{-\infty})\\
      &=-\big[ Z_{s}^{\star\t} \nabla^2 \Phi_{\rm M}(R+\tau_{s}-R'-\tau_{s'}) Z_{s'}^\star \big]_{\alpha\beta} + O(|R-R'|^{-\infty}),
    \end{align*}
    the leading order in Theorem \ref{thm:insulators}.

\textit{(ii) Homogeneous contribution.} Without loss of generality, we may use the change of variables as in \textit{(i)} and assume $M=1$. Then, we have 
\[
    \int_{\mathbb R^3} \eta(q) 
    f_{\rm hom}(q) e^{i q\cdot R} \D{q}
    = \sum_{n=1}^N \int_{\mathbb R^3} \eta(q) g_n(q) e^{i q\cdot R} \D{q},
    \qquad g_n(q) \coloneqq 
    \frac
        {p_n(q)}
        {|q|^{2N}}
\]
(recall that $p_n$ a homogeneous polynomial of degree $n + 2N$, so that
$g_{n}$ is homogeneous of degree $n$).

Notice $g_n$ is locally integrable and thus $\widecheck{g_n}$ is a tempered distribution. We wish to write $\widecheck{g_n}$ as a derivative of $[|q|^{-2N}]^\vee$, but since $|q|^{-2N}$ is not locally integrable (for all $N\geq 2)$, some care is needed.

\begin{lemma}\label{lem:g}
    For $n \ge 1$, $\widecheck{g_n}(x) = \alpha \left(p_n(-i\nabla)|\cdot|^{2N-3}\right)(x)$ for some
    $\alpha \in \R$.
\end{lemma}

The conclusion of \textit{(ii)} follows from this lemma. 
Let $h \in C_c^\infty(\R^3)$ be a real-space smooth non-negative radial cut-off function such that $h=1$ on $B_{|R|/2}$ and
$\supp h \subset B_{3|R|/4}$. Since $p_n$ has degree $n+2N$, we have
$\left|(p_n(-i\nabla) |\cdot|^{2N-3})(x)\right| \lesssim |x|^{-(n+3)}$ at infinity and thus
\begin{align}
    \Big| \int_{\R^3} &g_{n}(q) \eta(q) e^{i q \cdot R} \D q \Big| 
    = \left| c_N\Big\langle
    p_n(-i\nabla) |\cdot|^{2N-3}, \widecheck{\eta}(\cdot - R)
    \Big\rangle_{{\mathcal S}' \times {\mathcal S}} \right|\nonumber \\
    &\leq c_N
    \int_{\mathbb R^3}
    \left| (p_n(-i\nabla) |\cdot|^{2N-3})(x)\right|  (1-h(x))  \left|\widecheck{\eta}(x  - R)\right| \D x
    + \left|\big\langle \widecheck{g_n}, h \widecheck{\eta}(\,\cdot - R) \big\rangle_{{\mathcal S}' \times {\mathcal S}}\right| \nonumber\\
    &\lesssim \int\limits_{B_{|R|/2}} \!\!\!
    \frac
        {\left|\widecheck{\eta}(x  - R) \right|}
        {|x|^{n+3}}
    \D{x}
    + |R|^{-n}
    \sup_{ |y| > |R|/2 } \left|\widecheck{\eta}(y) \right|
    + \mathcal O(|R|^{-\ell}) 
    = \mathcal O(|R|^{-(n+3)}) \nonumber
\end{align} 
for all $\ell \in \mathbb N$, as $|R|\to \infty$, which is lower order since $n \ge 1$.

\begin{proof}[Proof of Lemma~\ref{lem:g}]
  The Cauchy principal value $\pv |\cdot|^{-2N}$ is the tempered distribution in $\mathcal S'(\mathbb R^3)$ such that for any $j \in \N$ with $j > 2N-4$,
  by
\begin{align}
    \Big\langle \pv \frac{1}{|\cdot|^{2N}}, \phi \Big\rangle_{{\mathcal S}' \times {\mathcal S}}
    &\coloneqq 
    \int_{B_1}
    \frac
        { \phi(q) - P_j[\phi](q)}
        {|q|^{2N}} 
    \D{q}
    + 
    \int_{B_1^c}
    \frac{\phi(q)}{|q|^{2N}}  \D{q} 
    %
    %
    + \sum_{
        \above
            {|\alpha|\leq j :}
            {\alpha_i \text{ even}}
    }
    \frac
        {\int_{\mathbb S^{2}} \theta^\alpha \D{\theta}}
        {|\alpha| + 3 - 2N}
     \frac
        {\partial^\alpha \phi(0)}
        {\alpha!}. 
    \nonumber
\end{align}
where $P_j[\phi]$ is the degree $j$ Taylor polynomial of $\phi$ at $0$. Here, one integrates the Taylor polynomial over $B_1$ using polar coordinates and notes that the odd powers vanish by symmetry arguments \cite{Grafakos2016}. One may check that for all $j > 2N-4$, the above formula actually defines a tempered distribution independent of the choice of $j$. Furthermore, using \cite[Theorem~2.4.6]{Grafakos2016}, we have $\widecheck{\pv |\cdot|^{-2N}} = c_N |x|^{2N-3}$ for some constant $c_N$.

Since $\pv |\cdot|^{-2N} = |\cdot|^{-2N}$ on
$\big\{ \phi \in \mathcal S(\mathbb R^d) \colon \partial^\alpha \phi(0) = 0 \,\, \forall |\alpha| < 2N-2\big\}$,
$g_n = p_n \pv |\cdot|^{-2N}$, and the result follows. 
\end{proof}

\begin{remark}
    More generally, one may define $u_z \coloneqq \Gamma( \frac{d+z}{2} )^{-1} \pv |\cdot|^{z} \in \mathcal S'(\mathbb R^d)$ for all $z \in \mathbb C$ with $\widehat{u_z} = c_{d,z} u_{-d-z}$ for some constant $c_{d,z}$. The proof follows from showing that $z\mapsto \Braket{u_z,\phi}$ is entire for all $\phi \in \mathcal S(\mathbb R^d)$ and computing $\widehat{u_z}$ for $-d < \re z < -d + \frac12$ explicitly \cite{Grafakos2016}.  
\end{remark}

\section*{Acknowledgements}
This project has received funding from the European Research Council (ERC) under the European Union’s Horizon 2020 research and innovation programme (grant agreement EMC2 No 810367) and the Simons Targeted Grant Award No. 896630.

\bibliographystyle{siam}
\bibliography{refs.bib}

\begin{thebibliography}{10}

\bibitem{Adler1962}
{\sc S.~L. Adler}, {\em Quantum theory of the dielectric constant in real solids}, Phys. Rev., 126 (1962), pp.~413--420.

\bibitem{Benzi2013}
{\sc M.~Benzi, P.~Boito, and N.~Razouk}, {\em Decay properties of spectral projectors with applications to electronic structure}, {SIAM} Review, 55 (2013), pp.~3--64.

\bibitem{BowlerMiyazaki2012}
{\sc D.~R. Bowler and T.~Miyazaki}, {\em {$O(N)$} methods in electronic structure calculations}, Rep. Prog. Phys., 75 (2012).

\bibitem{Cances2008}
{\sc {\'{E}}.~Canc{\`{e}}s, A.~Deleurence, and M.~Lewin}, {\em A new approach to the modeling of local defects in crystals: The reduced {H}artree--{F}ock case}, Communications in Mathematical Physics, 281 (2008), pp.~129--177.

\bibitem{Cances2013-disorder}
{\sc {\'E}.~Canc{\`e}s, S.~Lahbabi, and M.~Lewin}, {\em Mean-field models for disordered crystals}, J. Math. Pures Appl., 100 (2013), pp.~241--274.

\bibitem{Cances2010}
{\sc {\'{E}}.~Canc{\`{e}}s and M.~Lewin}, {\em The dielectric permittivity of crystals in the reduced {H}artree{\textendash}{F}ock approximation}, Arch. Rat. Mech. Anal., 197 (2009), pp.~139--177.

\bibitem{ChenLuOrtner18}
{\sc H.~Chen, J.~Lu, and C.~Ortner}, {\em Thermodynamic limit of crystal defects with finite temperature tight binding}, Arch. Ration. Mech. Anal., 230 (2018), pp.~701--733.

\bibitem{ChenNazarOrtner19}
{\sc H.~Chen, F.~Nazar, and C.~Ortner}, {\em Geometry equilibration of crystalline defects in quantum and atomic descriptions}, Math. Model. Methods Appl. Sci., 29 (2019), pp.~419--492.

\bibitem{ChenOrtner16}
{\sc H.~Chen and C.~Ortner}, {\em {QM/MM} methods for crystalline defects. {P}art 1: Locality of the tight binding model}, Multiscale Model. Simul., 14 (2016), pp.~232--264.

\bibitem{ChenOrtner17}
\leavevmode\vrule height 2pt depth -1.6pt width 23pt, {\em {QM/MM} methods for crystalline defects. {P}art 2: Consistent energy and force-mixing}, Multiscale Model. Simul., 15 (2017), pp.~184--214.

\bibitem{ChenOrtnerWang22}
{\sc H.~Chen, C.~Ortner, and Y.~Wang}, {\em {QM/MM} methods for crystalline defects. {P}art 3: {M}achine-learned interatomic potentials}, SIAM Multiscale Model. Simul., 20 (2022).

\bibitem{CsanyiAlbaretMorasPayneDeVita05}
{\sc G.~Cs{\'{a}}nyi, T.~Albaret, G.~Moras, M.~C. Payne, and A.~{De Vita}}, {\em Multiscale hybrid simulation methods for material systems}, J. Phys. Condens. Matter, 17 (2005), pp.~691--703.

\bibitem{EhrlacherOrtnerShapeev16}
{\sc V.~Ehrlacher, C.~Ortner, and A.~V. Shapeev}, {\em Analysis of boundary conditions for crystal defect atomistic simulations}, Arch. Ration. Mech. Anal., 222 (2016), pp.~1217--1268.

\bibitem{GiulianiVignale2005}
{\sc G.~Giuliani and G.~Vignale}, {\em Quantum Theory of the Electron Liquid}, Cambridge University Press, 2005.

\bibitem{Giuliani2005}
{\sc G.~F. Giuliani, G.~Vignale, and T.~Datta}, {\em {RKKY} range function of a one-dimensional noninteracting electron gas}, Physical Review B, 72 (2005).

\bibitem{Goedecker1998}
{\sc S.~Goedecker}, {\em Decay properties of the finite-temperature density matrix in metals}, Physical Review B, 58 (1998), p.~3501–3502.

\bibitem{Goedecker99}
{\sc S.~Goedecker}, {\em Linear scaling electronic structure methods}, Rev. Mod. Phys., 71 (1999), pp.~1085--1123.

\bibitem{Grafakos2016}
{\sc L.~Grafakos}, {\em Classical Fourier Analysis}, Springer, 08 2016.

\bibitem{hunziker1986distortion}
{\sc W.~Hunziker}, {\em Distortion analyticity and molecular resonance curves}, Annales de l'IHP Physique th{\'e}orique, 45 (1986), pp.~339--358.

\bibitem{Kohn1996}
{\sc W.~Kohn}, {\em Density functional and density matrix method scaling linearly with the number of atoms}, Phys. Rev. Lett., 76 (1996), pp.~3168--3171.

\bibitem{Levitt2020}
{\sc A.~Levitt}, {\em Screening in the finite-temperature reduced {H}artree{\textendash}{F}ock model}, Arch. Rat. Mech. Anal., 238 (2020), pp.~901--927.

\bibitem{Lions2001}
{\sc P.-L. Lions, I.~Catto, and C.~Le~Bris}, {\em On the thermodynamic limit for {Hartree--Fock} type models}, Ann. Inst. H. Poincare Anal. Non Lineaire, 18 (2001), pp.~687--760.

\bibitem{NazarOrtner17}
{\sc F.~Q. Nazar and C.~Ortner}, {\em Locality of the {T}homas--{F}ermi–von {W}eizs\"{a}cker equations}, Arch. Rat. Mech. Anal, 224 (2017), p.~817–870.

\bibitem{nier1993schrodinger}
{\sc F.~Nier}, {\em Schr{\"o}dinger--{P}oisson systems in dimension $d \leq 3$: The whole-space case}, Proceedings of the Royal Society of Edinburgh Section A: Mathematics, 123 (1993), pp.~1179--1201.

\bibitem{Niklasson2011}
{\sc A.~M.~N. Niklasson}, {\em Density matrix methods in linear scaling electronic structure theory}, in Linear-Scaling Tech. Comput. Chem. Phys., Springer, Dordrecht, 2011, ch.~16, pp.~439--473.

\bibitem{OrtnerThomas2020:pointdef}
{\sc C.~Ortner and J.~Thomas}, {\em Point defects in tight binding models for insulators}, Mathematical Models and Methods in Applied Sciences, 30 (2020), pp.~2753--2797.

\bibitem{ChenOrtnerThomas2019:locality}
{\sc C.~Ortner, J.~Thomas, and H.~Chen}, {\em Locality of interatomic forces in tight binding models for insulators}, {ESAIM}: Math. Model. Num. Anal., 54 (2020), pp.~2295--2318.

\bibitem{Prodan2005}
{\sc E.~Prodan, S.~R. Garcia, and M.~Putinar}, {\em Norm estimates of complex symmetric operators applied to quantum systems}, Journal of Physics A: Mathematical and General, 39 (2005), p.~389–400.

\bibitem{ProdanKohn05}
{\sc E.~Prodan and W.~Kohn}, {\em Nearsightedness of electronic matter}, Proc. Natl. Acad. Sci. U.S.A., 102 (2005), pp.~11635--11638.

\bibitem{Rubensson2011}
{\sc E.~H. Rubensson, E.~Rudberg, and P.~Salek}, {\em Methods for Hartree-Fock and Density Functional Theory Electronic Structure Calculations with Linearly Scaling Processor Time and Memory Usage}, Springer Netherlands, 2011, p.~263–300.

\bibitem{Simion2005}
{\sc G.~E. Simion and G.~F. Giuliani}, {\em Friedel oscillations in a {F}ermi liquid}, Physical Review B, 72 (2005).

\bibitem{Solovej1991}
{\sc J.~P. Solovej}, {\em Proof of the ionization conjecture in a reduced {H}artree--{F}ock model}, Inventiones mathematicae, 104 (1991), pp.~291--311.

\bibitem{Thomas2020:scTB}
{\sc J.~Thomas}, {\em Locality of interatomic interactions in self-consistent tight binding models}, J. Nonlinear Sci., 30 (2020), pp.~3293--3319.

\bibitem{ThomasChenOrtner2022:body-order}
{\sc J.~Thomas, H.~Chen, and C.~Ortner}, {\em Body-ordered approximations of atomic properties}, Archive for Rational Mechanics and Analysis, 246 (2022), pp.~1--60.

\bibitem{Wiser1963}
{\sc N.~Wiser}, {\em Dielectric constant with local field effects included}, Phys. Rev., 129 (1963), pp.~62--69.

\bibitem{Yang1991}
{\sc W.~Yang}, {\em Direct calculation of electron density in density-functional theory}, Phys. Rev. Lett., 66 (1991), pp.~1438--1441.

\end{thebibliography}

\end{document}